\documentclass[12pt]{amsart}
\usepackage{amsfonts}

\usepackage[letterpaper]{geometry}

\usepackage{amssymb}
\def\N{\mathbb{N}}

\def\t{{\mathbb F^n_2}}
\def\bo{\{0,1\}^n}

\newtheorem{thm}{\bf Theorem}[section]
\newtheorem{lemma}[thm]{\bf Lemma}
\newtheorem{prop}[thm]{\bf Proposition}

\newtheorem{cor}[thm]{\bf Corollary}

\theoremstyle{definition}

\begin{document}

\title[On uncertainty inequalities related to subcube partitions and additive energy]{On uncertainty inequalities related to subcube partitions and additive energy}

 \author[Norbert Hegyv\'ari]{Norbert Hegyv\'ari}
 \address{Norbert Hegyv\'{a}ri, ELTE TTK,
E\"otv\"os University, Institute of Mathematics, H-1117
P\'{a}zm\'{a}ny st. 1/c, Budapest, Hungary and Alfr\'ed R\'enyi Institute of Mathematics, Hungarian Academy of Science, H-1364 Budapest, P.O.Box 127.}
 \email{hegyvari@renyi.hu}

\begin{abstract}
The additive energy plays a central role in combinatorial number theory. We show an uncertainty inequality which indicates how the additive energy of support of a Boolean function, its degree and subcube partition are related.

AMS 2010 Primary 11B30, 11L03, Secondary 11B75

\end{abstract}

 \maketitle

\section{Introduction and Motivation}

\vskip0.7cm


In the graph theory it is a well-known result among many others, that a given graph $(V(G), e(G))$ the cardinality of the maximum independent set $\alpha(G)$ and the maximum degree $d$  fulfils the relation $\alpha(G)\cdot (d+1)\geq |V(G)|$. This relation tells us that the maximum independent set and the maximum degree can not be small simultaneously. In mathematics there are examples like this where there is a bound of quantities. These types of phenomenons are said to be commonly {\it uncertainty inequalities}. 

In this paper we are looking for connections between parameters of Boolean functions and some parameters from the additive combinatorics.

A Boolean function is defined as a map $f: \{0,1\}^n\mapsto \{0,1\}$, other times it is used $f: \{-1,1\}^n\mapsto \{-1,1\}$ e.t.c (see [D]). We will consider the set $\{0,1\}^n$ as $\t$ with the usual addition on field. We convert all results to $f: \t \mapsto \{0,1\}$ which is what will be used.

One can consider a Boolean function $f$ as an indicator of the set $A=f^{-1}(1)$; i.e.
$$
f(x) = \left\{
     \begin{array}{lr}
       1 &  x \in A\\
       0 &  x \notin A
     \end{array}
\right.
$$
The influence of coordinate $i$ on $f$ is defined as $Inf_i(f)=Pr_{x\in \bo}[f(x)\neq f(x+e_i)]$, where $x$ is uniformly distributed over $\bo$, and $f(x+e_i)$ means that we change the $i^{th}$ coordinate to $1$ if $x_i=0$ and to $0$ if $x_i=1$ respectively. The total influence of $f$ is defined to be $I(f):=\sum_i Inf_i(f)$.

For a set $A\subseteq \t$ (and this notion is defined in all semigroups in a similar way), the {\it additive energy} of $A$ is defined as the number of quadruples  $(a_1,a_2,a_3,a_4)$ for which $a_1+a_2=a_3+a_4$, formally $$
E(A):=|\{(a_1,a_2,a_3,a_4)\in A^4: \ a_1+a_2=a_3+a_4\}|.
$$
Clearly $|A|^2\ll E(A)\ll |A|^3$ holds, since the quadruple $(a_1,a_2,a_1,a_2)$ is always a solution and given $a_1,a_2,a_3$ the term $a_4$ is uniquely determined by them. This notion is introduced by Terence Tao, and plays a central role in additive combinatorics. (see e.g. [TV]).

Let $h(x)$ be the binary entropy function defined by $h(x)=-x\log x-(1-x)\log(1-x)$.

A decision tree which computes a Boolean function $f$ determines a partition of the cube $\bo$, where for every element of a given part, the value of $f$ at each leaf is the same. 

The subcube of $\bo$ is a set of vectors in the form:
$$
C=\{(*,*,\dots, x_i{_1},*,\dots,x_i{_2},\dots x_i{_k}\dots,*): *\in\{0,1\}\},
$$
i.e. those vectors of $\bo$ in which there are $k$ fix coordinates ($ x_i{_1},\dots,x_i{_2},\dots x_i{_k}$), and the rest are free. The dimension of this subcube is $2^{n-k}$.

Clearly there is a partition of $\bo$ into the union of subcubes $\cup_i C_i$, such that the value of the function $f$ is the same on each vector of $C_i$, i.e. for every $i$ and $x,y\in C_i$, $f(x)=f(y)$. 

For example when $f$ is a dictator function, i.e. $f(x_1,x_2,\dots, x_n)=x_i$ for some $1\leq i\leq n$ there are at most two subcubes.

However there exist a monochromatic subcube partition of $\bo$ which does not induce any decision tree; one of the simplest example is the quarternary majority function {\bf 4-Maj}: $\{0,1\}^4\mapsto \{0,1\}$ (see details e.g. in [KDS]).

Let us denote by $H_{scp}(f)$ the minimum number of subcubes in a subcube partition which computes the Boolean function $f$.

\medskip

\subsection{Prior work}

In the last decades there are several interplay between complexity theory and additive combinatorics. One of the most interesting example is connection between notions in computer sciences and the {\it Gowers norm} (see e.g. [ST], [TR]).
Another interesting example is an additive communication complexity problem which is supported by an example of Behrend on the maximal density of a set not containing three-term arithmetic progression (see e.g. [RY]). 

\bigskip

\section{Result}

The aim of this note is to prove the following uncertainty estimation related to $deg f$, the degree of $f$, $H_{scp}(f)$ and the additive energy $E(A)$:

\begin{thm}
Let $f$ be any Boolean function, $f: \bo \mapsto \{0,1\}$ with degree $deg f$, the set $A$ its support i.e. $A:=f^{-1}(1)$, $H=H_{scp}(f)$, and $E(A)$ its additive energy. We have the following uncertainity bound
$$
2^{3n}n^3\leq (8^{deg f}deg^2f)H^2\cdot E(A).
$$
\end{thm}

\section{Preliminaries}

\medskip

Let $f,g$ be two Boolean functions. The expected value of $f$ is
$$
\mathbb{E}(f):=\frac{1}{2^n}\sum_{x\in \bo}f(x),
$$
and the inner product of $f$ and $g$ is $\langle f,g\rangle:=\mathbb{E}(fg)$. For $S\subseteq [n]$ the corresponding  input is $x=(x_1,x_2,\dots,x_n)\in \bo$ namely $x_i=1$ if $i\in S$ and $x_i=0$ otherwise. A basis function or character is defined by $\chi_x(y):=(-1)^{\langle x,y\rangle_2}$, where $\langle x,y\rangle_2:=\sum_{i=1}^nx_iy_i \pmod 2$.

For a set $S\subseteq [n]$ the Fourier transform of $f$ is $\widehat{f}(S)=\langle f,\chi_S\rangle$. 

For the Fourier transform the following are true:
$$
(i)\  \langle f,g\rangle =\sum_{r\in \bo}\widehat{f}(r)\widehat{g}(r) \ \text{(Plancherel)}
$$
$$
(ii) \ \|f\|_2^2=\mathbb{E}(f^2)=\langle f,f\rangle = \sum_{r\in \bo}\widehat{f}^2(r) \ \text{(Parseval)}
$$
So by the Parseval formula for the indicator function we have $\sum_{r\in \bo}\widehat{A}^2(r)=\frac{1}{2^n}|A|$.


For functions $f$ and $g$ their convolution is defined by
$$
f\ast g (x):=\mathbb{E}f(y)g(x+y).
$$
It is easy to verify that the convolution is associative: $f\ast (g\ast h)=(f\ast g)\ast h$.

\smallskip

We will use the notation $|X|\ll |Y|$ to denote the estimate $|X|\leq C|Y|$ for some absolute constant $C>0$.

\bigskip

\section{Proof}

For the proof we  need some lemmas.

\begin{lemma}
\begin{equation}\label{1}
\mathbb{E}_{x,y,z}(f(x)f(y)f(z)f(x+y+z))=\sum_r\widehat{f}^4(r).
\end{equation}
\end{lemma}

This statement can be found for example in [D, p.22] without proof; so for the sake of completeness we include a short proof.

\begin{proof}

Using that $\mathbb{E}_z(f(z)f(x+y+z))=f\ast f(x+y)$, we have
$$
\mathbb{E}_{x,y,z}(f(x)f(y)f(z)f(x+y+z))=\mathbb{E}_x(f(x)\mathbb{E}_y(f(y)\mathbb{E}_z(f(z)f(x+y+z))))=
$$
$$
=\mathbb{E}_x(f(x)\mathbb{E}_y(f(y)f\ast f(x+y)))=\mathbb{E}_x(f(x)(f\ast(f\ast f(x)))).
$$
Write briefly $f_3\ast (r)$ instead of $(f\ast (f\ast f))(r)$. By the Plancherel formula, the associative of the convolution, and the Fourier transformation of a convolution we have
$$
\mathbb{E}_{x,y,z}(f(x)f(y)f(z)f(x+y+z))=\sum_r[\widehat{f\cdot f_3\ast (r)}]=
$$
$$
=\sum_r\widehat{f}(r)\widehat{f_3\ast (r)}=\sum_r\widehat{f}(r)\widehat{f}^3(r)=\sum_r\widehat{f}^4(r).
$$
\end{proof}

\begin{cor}\label{4.2}
Let $A:=f^{-1}(1)$. Then $\sum_r\widehat{f}^4(r)=\frac{1}{2^{3n}}E(A)$.
\end{cor}
\begin{proof}
As we detected $\sum_r\widehat{f}^4(r)$ can be written as $\mathbb{E}_{x,y,z}(f(x)f(y)f(z)f(x+y+z))$. Since in $\mathbb{F}_2$ $x+(x+y+z)=y+z$ holds, thus we have
$$
\mathbb{E}_{x,y,z}(f(x)f(y)f(z)f(x+y+z))=
$$
$$
=\frac{1}{2^{3n}}|\{(a_1,a_2,a_3,a_4)\in A^4: \ a_1+a_2=a_3+a_4\}|=\frac{1}{2^{3n}}E(A).
$$
\end{proof}

\medskip

The key step of the proof is to give a lower and an upper bound for the total influence. 

First recall that
$$
I(f)=\sum_{S\in [n]}|S|\widehat{f}^2(S).
$$
which can easily be proven.
\begin{lemma}\label{4.3}
$$
I(f)\leq \frac{1}{2^n}(deg f)^{2/3}\|f\|^{2/3}_1(E(A))^{1/3}.
$$
\end{lemma}
\begin{proof}
By the H\"older inequality
$$
I(f)=\sum_{S\in [n]}|S|\widehat{f}(S)^2=\sum_{S\in [n]}(|S|\widehat{f}(S))^{2/3}|\widehat{f}(S)|^{4/3}\leq
$$
$$
\leq \left(\sum_{S\in [n]}|S||\widehat{f}(S)|\right)^{2/3}\left(\sum_{S\in [n]}|\widehat{f}(S)|^4\right)^{1/3}\leq
$$
Now using Corollary \ref{4.2} and $(\sum_{S\in [n]}|S||\widehat{f}(S))^{2/3}\leq (deg f)^{2/3}(\sum_{S\in [n]}|\widehat{f}(S)|)^{2/3}$ , we have
$$
\leq \frac{1}{2^n}(deg f)^{2/3}\|f\|^{2/3}_1(E(A))^{1/3}
$$
as we stated.
\end{proof}

The lower bound for the total influence comes from the folklore; since $Inf_i(f)=Pr_{x\in \bo}[f(x)\neq f(x+e_i)]$, using Schwartz-Zippel lemma one can show, that $Inf_i(f)\geq \frac{1}{2^{deg f}}$ and hence
\begin{equation}\label{2}
I(f)\geq \frac{n}{2^{deg f}}
\end{equation}
(Maybe the first explicit estimation can be found in [NSZ]). 

\medskip

In the rest of the paper we recall some behaviour of the subcube partition to complete the proof of the theorem.

Now let $C_1,C_2,\dots, C_H$ be a minimal subcube partition of $\bo$ which computes $f$. Let us denote by $f_i$ the value of $f$ in the part $C_i$. So if $\bf{1}_i$ is the indicator function of $C_i$ then clearly $f(x)=\sum_{i=1}^Hf_i\bf{1}_i(x)$ and hence by the linearity we have 
$$
\widehat{f}(S)=\sum_{i=1}^Hf_i\widehat{\bf{1}_i(S)}.
$$
It is well-known that if for a function $g(x)\leq 1$ for all $x\in\bo$ holds then $\widehat{g}(S)\leq 1$ also holds. We will show that it is also true for the restricted indicator function as well. 

Now recall the simple fact that the equivalent form of the Fourier transform is also true for any subset $U\subseteq \bo$
$$
|\widehat{\bf{1}_U(S)}|=\Big|\frac{1}{2^n}\sum_T\bf{1}_U(T)(-1)^{|S\cap T|}\Big|\leq \frac{1}{2^n}\sum_{T}|\bf{1}_U(T\cap U)(-1)^{|S\cap T|}|\leq
$$
\begin{equation}\label{3}
\leq \frac{1}{2^n}\sum_{T\subseteq U}2^{n-|U|}=1.
\end{equation}
Thus using this bound for sets $U=C_i; \ i=,2,\dots H$ we have
$$
\|f\|_1=\sum_S\big|\widehat{f}(S)\big|\leq \sum_{i=1}^H|f_i\widehat{\bf{1}_i(S)}|\leq H.
$$
Finally by (\ref{2}), Lemma \ref{4.3} and the calculation above we get
$$
\frac{n}{2^{deg f}}\leq I(f)\leq \frac{1}{2^n}(deg f)^{2/3}\|f\|^{2/3}_1(E(A))^{1/3}\leq \frac{1}{2^n}(deg f)^{2/3}H^{2/3}(E(A))^{1/3}.
$$
Comparing the LHS and RHS and rearranging the inequality we obtain the desired estimation.

\section{Concluding remarks}

1. Let us first remark that there is a refinement of the theorem if we have an information on the cardinality of the subcubes. For instance when the cardinalities are concentrated to the "middle size": assume, there are parameters $\eta, \nu \in (0,1)$, such that $\eta \leq |C_i|/n\leq \nu $ holds for every $i=1,2,\dots H$. Then the bound for the Fourier transform of the indicators instead of (\ref{3}) will be
$$
|\widehat{\bf{1}_i(S)}|\leq \frac{1}{2^n}2^{|C_i|}\cdot \sum_{k\leq n-\eta n}{n\choose k}.
$$
Now using the bound for sets $C_i$ and using the well-known estimates for binomial coefficients (where $h(x)$ is the binary entropy function defined by $h(x)=-x\log x-(1-x)\log(1-x)$)
$$
\sum_{k=0}^{\varepsilon n}{n\choose k}\leq 2^{n h(\varepsilon)} \quad \varepsilon\in [0,1]; \ \varepsilon n \in \N
$$
we obtain a stronger bound in the theorem. Namely the factor $H^{2/3}$ should change to 
$$
(2^{(\nu+h(1-\eta)-1)n}H)^{2/3}
$$
which is less than the original one for middle concentrated parts.
\medskip

2. The calculated bound at the end of the proof of theorem we achieved
$$
I(f)\leq \frac{1}{2^n}(deg f)^{2/3}H^{2/3}(E(A))^{1/3}.
$$
Now let us introduce the entropy type quantity $\mathcal{E}(g):=\mathbb{E}(g)\log(1/\mathbb{E}(g))$. 
A classical result of Harper, Bernstein, Lindsey and Hart says $I(g)\geq 2\mathcal{E}(g)$ (see e.g. [KF]).

In $\mathcal{E}(g)$, $g=\mu(A)$ is the density  of the set $A$, i.e. one can read this entropy as $\mathcal{E}(\mu(A)):=\mathbb{E}(\mu(A))\log(1/\mathbb{E}(\mu(A)))$.

So one can conclude the following uncertainity inequality too:
\begin{prop}
$$
2^{3n+3}\mathcal{E}^3(\mu(A))\leq (deg f)^2H^2E(A).
$$
\end{prop}
\medskip

\noindent{\bf Acknowledgement.} This work is supported by NKFIH (OTKA) grant K-129335.

\end{document}